\documentclass[11pt]{scrartcl}

\usepackage{amsmath}
\usepackage{amssymb}
\usepackage{amsthm}

\newtheorem{theorem}{Theorem}
\newtheorem{lemma}{Lemma}
\newtheorem{proposition}{Proposition}

\newtheorem{definition}{Definition}
\newtheorem{example}{Example}

\title{Assigning a Small Agreeable Set of Indivisible Items to Multiple Players\footnote{A preliminary version of this paper appeared in Proceedings of the 25th International Joint Conference on Artificial Intelligence, July 2016.}}
\author{
Warut Suksompong\\Stanford University
}\date{\vspace{-5ex}}

\begin{document}
	
\maketitle

\begin{abstract}
We consider an assignment problem that has aspects of fair division as well as social choice. In particular, we investigate the problem of assigning a small subset from a set of indivisible items to multiple players so that the chosen subset is \emph{agreeable} to all players, i.e., every player weakly prefers the chosen subset to any subset of its complement. For an arbitrary number of players, we derive tight upper bounds on the size for which a subset of that size that is agreeable to all players always exists when preferences are monotonic. We then present polynomial-time algorithms that find an agreeable subset of approximately half of the items when there are two or three players and preferences are responsive. Our results translate to a 2-approximation on the individual welfare of every player when preferences are subadditive.
\end{abstract}

\section{Introduction}

Consider the following assignment problem with aspects of fair division as well as social choice: A group of agents are going together on a road trip, and they have to collectively agree upon a subset of items to bring with them in a bag that can hold a limited number of items. Naturally, the agents have different interests on different items, but they always like a set at least as much as any subset of it. The set of items included in the bag should be selected in such a way that any agent weakly prefers it to the set of items left at home. How large does the bag need to be in order to guarantee that there always exists such a subset---an \emph{agreeable} subset---that fits within it? And how can we efficiently compute an agreeable subset?

Our notion of agreeability is closely related to \emph{envy-freeness}, a well-studied notion in fair division \cite{BouveretLe14,BramsFeMo12,ChevaleyreDuEn06,CohlerLaPa11,KashPrSh13}, and variants of agreeability have been considered in the literature \cite{AzizGaMa15,BouveretEnLa10,BramsKiKl12}. An envy-free allocation is one in which every player weakly prefers her bundle to that of any other player. Both agreeability and envy-freeness involve players being happy if their allocation is at least as good as some other allocation(s). In previous work on fair division, each allocation is typically assigned to a single player, and the player compares the allocation that she receives with the allocations that other players receive. A novelty in our setting is the introduction of a social choice aspect to the allocation problem, i.e., we assign the same subset of items to multiple players with possibly wildly differing preferences. This corresponds to the social choice framework of aggregating individual preferences to reach a collective decision for the whole group \cite{ArrowSeSu02,BrandtCoEnLaPr16}.\footnote{In fact, fair division can also be viewed as a social choice problem in the sense that a collective decision, concerning the allocation of the resources to the different players, has to be reached.} However, the envy-freeness condition is one-directional in our setting---there is no envy associated to the unassigned set of items. Our setting can also be seen as a multiagent knapsack problem in which every item has the same size and the agreeability condition is imposed, and it bears a resemblance to the consensus halving and necklace bisection problems \cite{AlonWe86,SimmonsSu03}.


For an arbitrary number of players, we derive tight upper bounds on the size for which a subset of that size that is agreeable to all players always exists when preferences are monotonic. Remarkably, even if the subset of items needs to be allocated to several players, the number of extra items that we might need to choose in order to accommodate all of them is quite small, i.e., half an item per additional player. We then present polynomial-time algorithms that compute an agreeable subset of approximately half of the items when there are two or three players and preferences are responsive. Our results translate to a 2-approximation on the individual welfare of every player when preferences are subadditive.

The assumptions that we make on the players' preferences, monotonicity and responsiveness, are practical in a broad range of settings. Monotonicity says that a player cannot be worse off whenever an item is added to her set, while the stronger notion of responsiveness says that a player cannot be worse off whenever an item is added to her set or replaced by another item that she weakly prefers to the original item. These assumptions are common in the literature\footnote{For a comprehensive treatment of properties concerning the ranking of sets of objects, we refer to \cite{BarberaBaPo04}.} and have been made, for example, in \cite{BramsFi00,BramsKiKl12,Aziz15,LiptonMaMo04}, with the latter two only assuming monotonicity. For our existence results we will only assume monotonicity, and we will assume responsiveness for our algorithms. 

Even though the need to assign as many as half of the items even for two or three players might not seem that impressive at first glance, the reader should bear in mind that the conditions that we impose on the valuation functions are not many. For instance, monotonicity and responsiveness are generalizations of \emph{additivity}, a very common assumption on preferences (e.g., \cite{DickersonGoKa14}). Since we consider very general valuation functions, it is natural that the size of the selected set is relatively large. Indeed, in the simple case where we want to assign items to a single player whose preference only depends on the number of items and who prefers more items than less, we already need to assign to her at least half of the items. In this light, our result that we only need to assign half an item extra in the worst case for each additional player is quite intriguing.

An important issue when we discuss algorithms for computing desirable allocations is how we represent the players' preferences. Since preferences on subsets, unlike preferences on single items, might not have a succinct representation, this would already disallow algorithms that run in polynomial time in the number of players and the number of items if the algorithm were required to read the whole preference input. To circumvent this problem, we instead assume the existence of a preference oracle. In particular, we allow the algorithm to make a polynomial number of queries to the preference oracle. In each query, the algorithm can specify a player and two subsets of items to the preference oracle, and the oracle reveals the preference of that player between the two subsets. The use of the oracle model is common in the computer science literature \cite{ArmstrongSaBo12,JainKoVa11,McmahanGoBl03}.


\section{Definitions and notation}
\label{sec:definitions}

In this section, we introduce the setting and give definitions and notation that we will use throughout this paper.

We consider $n$ players, numbered $1,2,\ldots,n$, who must collectively choose a subset of the set $S=\{x_1,x_2,\ldots,x_m\}$ of $m$ indivisible items. Denote by $\mathcal{S}$ the set of all subsets of $S$. Each player $i$ is endowed with a preference relation $\succeq_i$, a reflexive, complete, and transitive ordering over $\mathcal{S}$. Let $\succ_i$ denote the strict part and $\sim_i$ the indifference part of the relation $\succeq_i$. For items $x$ and $y$, we will sometimes abuse notation and write $x\succeq y$ to mean $\{x\}\succeq\{y\}$.

Next, we define two properties of preferences that we will consider in the paper. Both properties are standard in the literature. The first property, monotonicity, says roughly that ``more is better'', i.e., a player cannot be worse off whenever an item is added to her set.

\begin{definition}
A preference $\succeq$ on $\mathcal{S}$ is \emph{monotonic} if $T\cup\{x\}\succeq T$ for all $T\subseteq S$.
\end{definition}

Note that if $x\in T$, then $T\cup\{x\}\succeq T$ always holds, so we only need to check when $x\in S\backslash T$.

We will assume throughout most of the paper that preferences are monotonic. Monotonicity is a natural assumption in a broad range of situations. In particular, it implies free disposal of items, i.e., every item is considered to be of nonnegative value to each player.

The next property, responsiveness, says that a player cannot be worse off whenever an item is added to her set or replaced by another item that she weakly prefers to the original item. While stronger than monotonicity, responsiveness is still a reasonable assumption in many settings.

\begin{definition}
A preference $\succeq$ on $\mathcal{S}$ is \emph{responsive} if it satisfies the following two conditions:
\begin{itemize}
\item $\succeq$ is monotonic;
\item $T\backslash \{y\}\cup \{x\}\succeq T$ for all $T\subseteq S$ and $x,y$ such that $x\succeq y$, $x\not\in T$ and $y\in T$.
\end{itemize}
\end{definition}

Given two sets $T_1=\{x_{i_1},x_{i_2},\ldots,x_{i_k}\}$ with $i_1<i_2<\cdots<i_k$ and $T_2=\{x_{j_1},x_{j_2},\ldots,x_{j_k}\}$ with $j_1<j_2<\cdots<j_k$, $T_1$ is said to \emph{lexicographically dominate} $T_2$ if the following two statements hold for some $l\in\{1,2,\ldots,k\}$: 
\begin{enumerate}
\item $i_{l'}=j_{l'}$ for all $1\leq l'<l$; 
\item $i_l<j_l$.
\end{enumerate}

For instance, one can check that the set $\{x_1,x_3,x_5,x_8\}$ lexicographically dominates the set $\{x_1,x_3,x_6,x_7\}$, while the set $\{x_4,x_5\}$ does not lexicographically dominate the set $\{x_3,x_{10}\}$.

We now consider two examples of preferences. 

\begin{example}
\label{ex:numpref}
Consider the preference $\succeq_1$ on $\mathcal{S}$ defined as follows. For distinct subsets $T_1,T_2\subseteq S$, we have $T_1\succ_1 T_2$ if one of the following two conditions holds:
\begin{itemize}
\item $|T_1|>|T_2|$;
\item $|T_1|=|T_2|$ and $T_1$ lexicographically dominates $T_2$.
\end{itemize}

For instance, when $m=3$, the preference $\succeq_1$ is given by \[\{x_1,x_2,x_3\}\succ_1\{x_1,x_2\}\succ_1\{x_1,x_3\}\succ_1\{x_2,x_3\}\succ_1\{x_1\}\succ_1\{x_2\}\succ_1\{x_3\}\succ_1\emptyset.\]
\end{example}

\begin{example}
\label{ex:onepref}
Consider the preference $\succeq_2$ on $\mathcal{S}$ defined as follows. For distinct subsets $T_1,T_2\subseteq S$, we have $T_1\succ_2 T_2$ if one of the following three conditions holds:
\begin{itemize}
\item $x_1\in T_1$ and $x_1\not\in T_2$;
\item $x_1\in T_1,T_2$ or $x_1\not\in T_1,T_2$, and $|T_1|>|T_2|$;
\item $x_1\in T_1,T_2$ or $x_1\not\in T_1,T_2$, and $|T_1|=|T_2|$ and $T_1$ lexicographically dominates $T_2$.
\end{itemize}

For instance, when $m=3$, the preference $\succeq_2$ is given by \[\{x_1,x_2,x_3\}\succ_2\{x_1,x_2\}\succ_2\{x_1,x_3\}\succ_2\{x_1\}\succ_2\{x_2,x_3\}\succ_2\{x_2\}\succ_2\{x_3\}\succ_2\emptyset.\]
\end{example}

One can check that the preferences in Examples \ref{ex:numpref} and \ref{ex:onepref} are both monotonic and responsive. Moreover, they share the same preference on the single items in $S$, i.e., $x_1\succ x_2\succ\dots\succ x_m$.

Next, we define the notion of agreeability. We denote the complement of a set $T$ by $-T:=S\backslash T$.

\begin{definition}
\label{def:fiftypercent}
Player $i$ regards a subset $T\subseteq S$ as \emph{agreeable} if $T\succeq_i -T$.
\end{definition}

When preferences are monotonic, agreeability also implies that the player does not prefer any subset of the complement to her current set. That is, we have $T\succeq_i U$ for any $U\subseteq -T$.

We now define two properties of subsets of items based on the preference on single items. In general, we will use $\succeq$ to denote the preference on $\mathcal{S}$ and $\succeq^{sing}$ to denote the preference on the single items in $S$. The first property, possible agreeability, says that a subset is ``strictly agreeable'' for some responsive preference consistent with the preference on single items.

\begin{definition}
Fix a preference $\succeq^{sing}$ on the single items in $S$. A subset $T\subseteq S$ is \emph{possibly agreeable} with respect to $\succeq^{sing}$ if there exists a responsive preference $\succeq$ on $\mathcal{S}$ consistent with $\succeq^{sing}$ such that $T\succ -T$.
\end{definition}

The next property, necessary agreeability, concerns the situation in which a subset is agreeable for any responsive preference consistent with the preference on single items.

\begin{definition}
Fix a preference $\succeq^{sing}$ on the single items in $S$. A subset $T\subseteq S$ is \emph{necessarily agreeable} with respect to $\succeq^{sing}$ if $T\succeq -T$ for any responsive preference $\succeq$ on $\mathcal{S}$ consistent with $\succeq^{sing}$.
\end{definition}

We now make a connection to the model in which every player has a cardinal utility for each subset of items. A \emph{utility function} $f$ is a function that maps any subset of items to a nonnegative real number. We assume that utility functions are \emph{monotonic}, i.e., $f(T_1)\leq f(T_2)$ for all subsets of items $T_1\subseteq T_2$. A utility function $f$ is said to be \emph{additive} if $f(T_1\cup T_2)=f(T_1)+f(T_2)$ for all disjoint subsets of items $T_1$ and $T_2$, and is said to be \emph{subadditive} if $f(T_1\cup T_2)\leq f(T_1)+f(T_2)$ for all subsets of items $T_1$ and $T_2$. Any nonnegative monotonic additive utility function is also subadditive. Subadditive utility functions have been widely considered in the literature \cite{Feige09,BhawalkarRo11}.

When the preferences of the players are given by subadditive utility functions, a subset that a player regards as agreeable also gives the player a utility of at least $1/2$ of the utility that the player obtains from the whole set $S$ of items. Indeed, for any such subset $T$ we have $f(S)=f(T\cup -T)\leq f(T)+f(-T)\leq 2f(T)$, which implies that $f(T)\geq\frac{1}{2}f(S)$. Hence our results concerning subsets that are agreeable to every player also translate to a 2-approximation on the individual welfare of every player when preferences are subadditive.

\section{Possibly and necessarily agreeable subsets}
\label{sec:feasnecc}

In this section, we investigate possibly and necessarily agreeable subsets and provide characterizations for such sets. The characterizations will later be useful in the analyses of our algorithms for computing an agreeable subset when there are two or three players. The following proposition shows a relationship between the two classes of subsets.

\begin{proposition}
\label{prop:feasnecciff}
Fix a preference $\succeq^{sing}$ on the single items in $S$. A subset $T\subseteq S$ is possibly agreeable with respect to $\succeq^{sing}$ if and only if $-T$ is not necessarily agreeable with respect to $\succeq^{sing}$.
\end{proposition}

\begin{proof}
Suppose first that $T$ is possibly agreeable with respect to $\succeq^{sing}$. By definition, there exists a responsive preference $\succeq$ on $\mathcal{S}$ consistent with $\succeq^{sing}$ such that $T\succ -T$. This means that $-T\succeq T$ does not hold for the responsive preference $\succeq$. Hence $-T$ is not necessarily agreeable with respect to $\succeq^{sing}$.

The converse can be shown similarly.
\end{proof}

We now characterize possibly and necessarily agreeable subsets. 

\begin{proposition}
\label{prop:feasneccchar}
Fix a preference $\succeq^{sing}$ on the single items in $S$ with \[x_1\succeq^{sing}x_2\succeq^{sing}\cdots\succeq^{sing}x_m.\] Let $T\subseteq S$, and define $I_k=\{x_1,x_2,\ldots,x_k\}$ for all $k\in\{1,2,\ldots,m\}$.

(a) If $|I_k\cap T|\geq\frac{k}{2}$ for all $k\in\{1,2,\ldots,m\}$, then $T$ is necessarily agreeable with respect to $\succeq^{sing}$. The converse also holds if the preference $\succeq^{sing}$ is strict.

(b) If $T$ is possibly agreeable with respect to $\succeq^{sing}$, then $|I_k\cap T|>\frac{k}{2}$ for some $k\in\{1,2,\ldots,m\}$. The converse also holds if the preference $\succeq^{sing}$ is strict.
\end{proposition}

\begin{proof}
(a) Suppose first that $|I_k\cap T|\geq\frac{k}{2}$ for all $k\in\{1,2,\ldots,m\}$. Since $|I_m\cap T|\geq\frac{m}{2}$, we have that $|T|\geq|-T|$. Let $T'\subseteq T$ be the subset consisting of the $|-T|$ elements of $T$ with the smallest indices. 

Define a bijective function $f:T'\rightarrow -T$ as follows: Given the element $x_k\in T$ with the smallest index for which $f(x_k)$ is not yet defined, we define $f(x_k)$ to be the element in $-T$ with the smallest index that has not occurred in the range of $f$ so far. Since $|I_k\cap T|\geq\frac{k}{2}$ for all $k\in\{1,2,\ldots,m\}$, the function $f$ maps each element $x_k$ to another element $x_l$ with $l>k$. Hence $T$ is necessarily agreeable with respect to $\succeq^{sing}$.

Now, suppose that the preference $\succeq^{sing}$ is strict, and that $|I_k\cap T|<\frac{k}{2}$ for some $k\in\{1,2,\ldots,m\}$. Since the $(\lfloor\frac{k}{2}\rfloor+1)$-th element of $-T$ has a smaller index than the $(\lfloor\frac{k}{2}\rfloor+1)$-th element of $T$ (if the latter element exists at all), there exists a responsive preference $\succeq$ on $\mathcal{S}$ consistent with $\succeq^{sing}$ such that $-T\succ T$. Hence $T$ is not necessarily agreeable with respect to $\succeq^{sing}$.

(b) Suppose first that $T$ is possibly agreeable with respect to $\succeq^{sing}$. By Proposition \ref{prop:feasnecciff}, $-T$ is not necessarily agreeable with respect to $\succeq'$. Part (a) implies that $|I_k\cap -T|<\frac{k}{2}$ for some $k\in\{1,2,\ldots,m\}$. It follows that $|I_k\cap T|>\frac{k}{2}$ for this value of $k$.

When the preference $\succeq^{sing}$ is strict, the converse can be shown similarly to the converse of part (a).
\end{proof}

\section{Two players}
\label{sec:twoplayers}

In this section, we assume that there are two players and consider the problem of choosing a small subset that is agreeable to both players. We show that if preferences are monotonic, we can choose a subset of approximately half of the items that is agreeable to both players. Moreover, if we are given the preferences on the single items of the two players, we can choose a subset of approximately half of the items that is necessarily agreeable with respect to both preferences, and we present a polynomial-time algorithm to do so.

For arbitrary preferences that are not monotonic, the task of choosing any subset that is agreeable to two players might not be possible. Indeed, suppose that preferences are strict, and that the preference of one player is exactly the reverse of that of the other player. Then any subset is not agreeable for one of the players.

Nevertheless, when we impose the light assumption of monotonicity on the preferences, the following proposition shows that we can choose a subset of approximately half of the items that is agreeable to both players. Even though the proposition will later be generalized by Theorem \ref{thm:nplayers}, the proof of Theorem \ref{thm:nplayers} relies crucially on a theorem best known as Kneser's conjecture, which we use as a black box. We believe that a direct proof of Proposition \ref{prop:twoplayers} may well be useful both for understanding and for applying the result.

\begin{proposition}
\label{prop:twoplayers}
Assume that there are two players with monotonic preferences on $\mathcal{S}$. There exists a subset $T\subseteq S$ such that $|T|\leq\left\lceil\frac{m+1}{2}\right\rceil$ and $T$ is agreeable to both players. Moreover, there exist monotonic preferences for which the bound $\lceil\frac{m+1}{2}\rceil$ is tight.
\end{proposition}

\begin{proof}
Denote by $\succeq_1$ and $\succeq_2$ the preferences on $\mathcal{S}$ of the two players. 

Assume first that $m=2k+1$ is odd. Suppose for contradiction that no subset of size at most $k+1$ is agreeable to both players. Let $T\subseteq S$ be such that $|T|=k$. We begin by proving the following claim. 

\emph{Claim:} If $T\succ_1 -T$, then \[T\cup\{x\}\backslash\{x'\}\succ_1 -T\backslash\{x\}\cup\{x'\}\] for any $x\in -T$ and $x'\in T$. 

\emph{Proof of Claim:} Suppose that $T\succ_1 -T$, $x\in -T$, and $x'\in T$. It follows from monotonicity of the preferences that $T\cup\{x\}\succ_1 -T\backslash\{x\}$. Since no subset of size at most $k+1$ is agreeable to both players, we have $-T\backslash\{x\}\succ_2 T\cup\{x\}$. By monotonicity of the preferences again, we have $-T\backslash\{x\}\cup\{x'\}\succ_2 T\cup\{x\}\backslash\{x'\}$. Hence $T\cup\{x\}\backslash\{x'\}\succ_1 -T\backslash\{x\}\cup\{x'\}$, and our claim is proved. 

We now use our claim to obtain the desired contradiction. Assume without loss of generality that $\{x_1,x_2,\ldots,x_k\}\succ_1\{x_{k+1},x_{k+2},\ldots,x_{2k+1}\}$. Applying our claim repeatedly to move elements between the two sets, we find that $\{x_{k+1},x_2,\ldots,x_k\}\succ_1\{x_1,x_{k+2},\ldots,x_{2k+1}\}$, $\{x_{k+1},x_{k+2},x_3,\ldots,x_k\}\succ_1\{x_1,x_2,x_{k+3},\ldots,x_{2k+1}\}$, and so on, until finally $\{x_{k+1},x_{k+2},\ldots,x_{2k}\}\succ_1\{x_1,x_2,\ldots,x_k,x_{2k+1}\}$. By monotonicity of the preferences, we have that $\{x_{k+1},x_{k+2},\ldots,x_{2k+1}\}\succ_1\{x_1,x_2,\ldots,x_k\}$, which contradicts our assumption that $\{x_1,x_2,\ldots,x_k\}\succ_1\{x_{k+1},x_{k+2},\ldots,x_{2k+1}\}$.


Assume now that $m=2k$ is even. Let $S'$ be the set of all elements of $S$ except $x_1$. We know from the case of $m$ odd that there exists a subset $T\subseteq S'$ of size at most $k$ such that $T\succeq_1 S'\backslash T \text{ and } T\succeq_2 S'\backslash T$. Since the players have monotonic preferences, we have that $T\cup\{x_1\}\succeq_1 S'\backslash T \text{ and } T\cup\{x_1\}\succeq_2 S'\backslash T$. This means that the set $T\cup\{x_1\}$ of size at most $k+1$ is our desired subset.

Finally, we show that there exist preferences for which the bound $\lceil\frac{m+1}{2}\rceil$ is tight. Consider the preferences $\succeq_1$ and $\succeq_2$ on $\mathcal{S}$ defined as follows. For distinct subsets $T_1,T_2\subseteq S$, we have $T_1\succ_1 T_2$ if either $|T_1|>|T_2|$, or $|T_1|=|T_2|$ and $T_1$ lexicographically dominates $T_2$. Similarly, we have $T_1\succ_2 T_2$ if either $|T_1|>|T_2|$, or $|T_1|=|T_2|$ and $T_2$ lexicographically dominates $T_1$. One can check that the two preferences are monotonic, and that the bound $\lceil\frac{m+1}{2}\rceil$ is tight for these preferences.
\end{proof}


While Proposition \ref{prop:twoplayers} shows the existence of a subset of approximately half of the items that is agreeable to both players, it does not show how to obtain such a subset. If preferences are also responsive, it is possible to lift preferences over single items to partial preferences over subsets. The next theorem shows that by considering only preferences over single items, it is possible to find a subset of approximately half of the items that is guaranteed to be agreeable to both players, no matter what the (responsive) preferences are.

\begin{theorem}
\label{thm:twoplayers}
Assume that there are two players with preferences $\succeq^{sing}_1$ and $\succeq^{sing}_2$ on the single items in $S$. There exists a subset $T\subseteq S$ such that $|T|\leq\left\lceil\frac{m+1}{2}\right\rceil$ and $T$ is necessarily agreeable with respect to both $\succeq^{sing}_1$ and $\succeq^{sing}_2$. Moreover, there exist preferences on the single items in $S$ for which the bound $\lceil\frac{m+1}{2}\rceil$ is tight. We also give a polynomial-time algorithm that finds such a subset $T$.
\end{theorem}

\begin{proof}
Assume first that $m=2k+1$ is odd, and suppose without loss of generality that $x_1\succeq^{sing}_1 x_2\succeq^{sing}_1\cdots\succeq^{sing}_1 x_{2k+1}$. One can check using Proposition \ref{prop:feasneccchar} that the set $\{x_1,x_3,x_5,\ldots,x_{2k+1}\}$ is necessarily agreeable with respect to $\succeq^{sing}_1$. We choose our set $T$ of $k+1$ items as follows. 
\begin{enumerate}
\item We choose $x_1$. 
\item Between each of the $k$ pairs of items $\{x_2,x_3\},\{x_4,x_5\},\ldots,\{x_{2k},x_{2k+1}\}$, we choose the item that is preferred according to $\succeq^{sing}_2$. If $\succeq^{sing}_2$ is indifferent between any pair of items, we choose arbitrarily.
\end{enumerate}

Since our set $T$ is at least as good as the set $\{x_1,x_3,x_5,\ldots,x_{2k+1}\}$ with respect to $\succeq^{sing}_1$, and the latter set is necessarily agreeable with respect to $\succeq^{sing}_1$, $T$ is also necessarily agreeable with respect to $\succeq^{sing}_1$. Moreover, since we choose the item that $\succeq^{sing}_2$ prefers from each of the sets $\{x_2,x_3\},\{x_4,x_5\},\ldots,\{x_{2k},x_{2k+1}\}$, $T$ is also necessarily agreeable with respect to $\succeq^{sing}_2$. Hence $T$ is necessarily agreeable with respect to both $\succeq^{sing}_1$ and $\succeq^{sing}_2$.

Assume now that $m=2k$ is even. Let $S'$ be the set of all items of $S$ except $x_1$. We apply the algorithm from the case of $m$ odd to choose a set $T\subseteq S'$ of $k$ items such that $T$ is necessarily agreeable with respect to both $\succeq^{sing}_1$ and $\succeq^{sing}_2$, when the universe considered is $S'$. It follows that $T\cup\{x_1\}$ is a subset of $k+1$ items that is necessarily agreeable with respect to both $\succeq^{sing}_1$ and $\succeq^{sing}_2$, when the universe considered is $S$.

Next, we show that there exist preferences for which the bound $\lceil\frac{m+1}{2}\rceil$ is tight. If $m=2k+1$ is odd and the preferences in $\succeq^{sing}_1$ are strict, then by Proposition \ref{prop:feasneccchar}, any subset $T\subseteq S$ that is necessarily agreeable with respect to $\succeq^{sing}_1$ alone must already contain at least $k+1$ items. 

Now suppose that $m=2k$ is even, and let $\succeq^{sing}_1$ and $\succeq^{sing}_2$ be given by $x_1\succ^{sing}_1 x_2\succ^{sing}_1 \cdots\succ^{sing}_1 x_{2k}$ and $x_{2k}\succ^{sing}_2 x_{2k-1}\succ^{sing}_2 \cdots\succ^{sing}_2 x_{1}$. By Proposition \ref{prop:feasneccchar}, any subset $T\subseteq S$ that is necessarily agreeable with respect to $\succeq^{sing}_1$ alone must contain at least $k$ items, one of which is $x_1$. If $T$ contains exactly $k$ items, then it contains exactly $k-1$ items among $x_2,x_3,\ldots,x_{2k}$. Proposition \ref{prop:feasneccchar} implies that such a set $T$ is not necessarily agreeable with respect to $\succeq^{sing}_2$. Hence any subset $T\subseteq S$ that is necessarily agreeable with respect to both $\succeq^{sing}_1$ and $\succeq^{sing}_2$ must contain at least $k+1$ items, as desired.
\end{proof}


Since the algorithm in Theorem \ref{thm:twoplayers} only considers preferences on single items, it remains efficient even when the number of items is relatively large.

We can generalize the notion of necessary agreeability by defining a subset $T\subseteq S$ to be \emph{necessarily worth at least $1/k$} with respect to $\succeq^{sing}$ for some fixed positive integer $k\geq 2$ if the complement $-T$ can be partitioned into $k-1$ sets $T_1,T_2,\ldots,T_{k-1}$ so that $T\succeq T_i$ for all $i$ and any responsive preference $\succeq$ on $\mathcal{S}$ consistent with $\succeq^{sing}$. Necessary agreeability then corresponds to being necessarily worth at least $1/2$. An algorithm similar to that in Theorem \ref{thm:twoplayers} yields a subset $T\subseteq S$ such that $|T|\leq\left\lceil\frac{m+k-1}{k}\right\rceil$ and $T$ is necessarily worth at least $1/k$ with respect to both $\succeq^{sing}_1$ and $\succeq^{sing}_2$. This bound is again tight.

\section{Three or more players}
\label{sec:manyplayers}

In this section, we consider the problem of choosing a small subset that is agreeable to multiple players. We generalize Proposition \ref{prop:twoplayers} by showing that if preferences are monotonic, we can choose a subset of size approximately half of the number of items plus half of the number of players that is agreeable to all players. Moreover, if there are three players and preferences are also responsive, we present a polynomial-time algorithm for computing such a subset.

We begin with the following lemma.

\begin{lemma} 
\label{lemma:kneser}
The graph with all $k$-element subsets of $\{1,2,\ldots,n\}$ as vertices and with edges connecting disjoint sets has chromatic number $n-2k+2$ if $n\geq 2k$, and $1$ otherwise.
\end{lemma}

Lemma \ref{lemma:kneser} was conjectured by Kneser, and is therefore best known as Kneser's conjecture \cite{Kneser55}. It was first resolved by Lov\'{a}sz using topological methods, and was later simplified several times, before Matou\v{s}ek gave a purely combinatorial proof \cite{Lovasz78,Matousek04}. The lemma will play a crucial role in our next theorem.

Proposition \ref{prop:twoplayers} shows that if there are two players, we can choose a subset of approximately half of the items that is agreeable to both players. It is natural to ask what happens when there are more than two players. Given that the preferences of the players can differ wildly, it is perhaps surprising that the price we need to pay is only approximately half an item for each additional player, as the following theorem shows.

\begin{theorem}
\label{thm:nplayers}
Assume that there are $n$ players with monotonic preferences on $\mathcal{S}$. There exists a subset $T\subseteq S$ such that $|T|\leq\min\left(\left\lceil\frac{m+n-1}{2}\right\rceil,m\right)$ and $T$ is agreeable to all $n$ players. Moreover, there exist monotonic preferences for which the bound $\min\left(\lceil\frac{m+n-1}{2}\rceil,m\right)$ is tight.
\end{theorem}

\begin{proof}
Let $k=\lceil\frac{m+n-1}{2}\rceil$. If $k\geq m$, the set $S$ of all items has size $m=\min(k,m)$ and is agreeable to all $n$ players since preferences are monotonic. Assume from now on that $k<m$, and consider the graph $G$ with all $(m-k)$-element subsets of $\{x_1,x_2,\ldots,x_m\}$ as vertices and with edges connecting disjoint sets. If all $n$ players weakly prefer $S\backslash T$ to $T$ for some $(m-k)$-element subset $T\subseteq S$, then $S\backslash T$ is our desired subset of size $k$. 

Suppose now that for any $(m-k)$-element subset $T\subseteq S$, there exists a player who strictly prefers $T$ to $S\backslash T$. We color the vertices of $G$ with $n$ colors $1,2,\ldots,n$ in the following way. For each vertex $v$ of $G$ corresponding to a set $T$, we color it with the color corresponding to a player who strictly prefers $T$ to $S\backslash T$. If there is more than one such player, choose one arbitrarily. 

By Lemma \ref{lemma:kneser}, the chromatic number of $G$ is $m-2(m-k)+2=2k-m+2\geq n+1$. Since we colored $G$ with $n$ colors, there exist two adjacent vertices sharing the same color, say, $T_1$ and $T_2$. This means that $T_1\succ_i S\backslash T_1$ and $T_2\succ_i S\backslash T_2$ for some player $i$. Since $T_1$ and $T_2$ are disjoint, we have $T_1\subseteq S\backslash T_2$ and $T_2\subseteq S\backslash T_1$. Monotonicity of preferences now implies that $S\backslash T_1\succeq_i T_2\succ_i S\backslash T_2\succeq_i T_1$, a contradiction to $T_1\succ_i S\backslash T_1$. This case is thus impossible, which implies that we can always find our desired subset $T$.

Finally, we show that there exist monotonic preferences for which the bound $\min\left(k,m\right)$ is tight. We consider two cases.

\emph{Case 1}: $n\geq m$. Then $\min(k,m)=m$. For $i=1,2,\ldots,m$, let the preference of player $i$ be such that she prefers item $x_i$ to the subset $S\backslash\{x_i\}$ of the remaining items. Such a monotonic preference exists, e.g., the preference in Example \ref{ex:onepref}. Then any subset $T\subseteq S$ that is agreeable to player $i$ must contain item $i$. Hence a subset $T$ that is agreeable to all players must contain all $m$ items.

\emph{Case 2}: $n< m$. Then $\min(k,m)=k$. For $i=1,2,\ldots,n-1$, let the preference of player $i$ be such that she prefers item $x_i$ to the subset $S\backslash\{x_i\}$ of the remaining items. Such a monotonic preference exists, e.g., the preference in Example \ref{ex:onepref}. Let the preference of player $n$ depend solely on the presence of items $n,n+1,\ldots,m$. In particular, player $n$ prefers to have more of these items than less. 

For $i=1,2,\ldots,n-1$, any subset $T\subseteq S$ that is agreeable to player $i$ must contain item $i$. Also, any subset $T$ that is agreeable to player $n$ must contain at least half of the items $n,n+1,\ldots,m$. Hence a subset $T$ that is agreeable to all players must have size at least $n-1+\left\lceil\frac{m-n+1}{2}\right\rceil=k$, as desired.
\end{proof}


We again ask the question of how to obtain a subset whose existence is guaranteed by Theorem \ref{thm:nplayers}. The following theorem provides a polynomial-time algorithm in the case of three players with responsive preferences. Since preferences on subsets, unlike preferences on single items, might not have a succinct representation, we assume that the algorithm is allowed to make a polynomial number of queries to a preference oracle. In each query, the algorithm can specify a player and two subsets of items to the preference oracle, and the oracle reveals the preference of that player between the two subsets.

The algorithm presented in the theorem is weaker than the one for two players in Theorem \ref{thm:twoplayers} in that it relies not only on preferences on single items but also on preferences on subsets.

\begin{theorem}
\label{thm:threeplayers}
Assume that there are three players with responsive preferences $\succeq_1$, $\succeq_2$, and $\succeq_3$ on $\mathcal{S}$. There exists a subset $T\subseteq S$ such that $|T|\leq\left\lceil\frac{m}{2}\right\rceil+1$ and $T$ is agreeable for all three players. Moreover, there exist responsive preferences for which the bound $\lceil\frac{m}{2}\rceil+1$ is tight. We also give a polynomial-time algorithm that finds such a subset $T$. 
\end{theorem}

\begin{proof}
The existence of a subset $T\subseteq S$ such that $|T|\leq\lceil\frac{m}{2}\rceil+1$ and $T$ is agreeable for all three players follows from Theorem \ref{thm:nplayers}. The tightness of the bound also follows, since one can make the preferences used to show tightness in Theorem \ref{thm:nplayers} responsive. Next, we give an algorithm to find such a subset $T$.

Assume first that $m=2k$ is even. Suppose without loss of generality that $x_{2k-1}$ is the most preferred item according to $\succeq_1$, $x_{2k}$ is the most preferred item besides $x_{2k-1}$ according to $\succeq_2$, and that among the remaining $2k-2$ items, the preference $\succeq_1$ ranks them as $x_1\succeq_1 x_2\succeq_1\cdots\succeq_1 x_{2k-2}$. One can check using Proposition \ref{prop:feasneccchar} that the set $\{x_{2k-1},x_2,x_4,x_6,\ldots,x_{2k-2},x_{2k}\}$ is agreeable to player 1.

Let $A=\{x_1,x_2,\ldots,x_{2k-2}\}$, and consider the pairs $(x_1,x_2),(x_3,x_4),\ldots,(x_{2k-3},x_{2k-2})$. Let $B$ be a set of $k-1$ elements containing an item from each pair that is less preferred by $\succeq_2$ than its pair. If $\succeq_2$ prefers the two items in some pair equally, we choose an arbitrary item from that pair.

While $A\backslash B\succeq_2 B$, we remove an element from $B$ that is less preferred by $\succeq_2$ than its pair, and insert its pair into $B$. Since the preference is responsive and there are finitely many items, we eventually reach a point where $B\succeq_2 A\backslash B$. We consider two cases.

\emph{Case 1}: We have not performed any insertion into or removal from $B$. By definition of $B$, we have that $A\backslash B\succeq_2 B$, and therefore $A\backslash B\sim_2 B$. Since preferences are monotonic, it follows that $A\backslash B\cup\{x_{2k}\}\succeq_2 B$ and $B\cup\{x_{2k}\}\succeq_2 A\backslash B$.

\emph{Case 2}: We have performed some insertion into or removal from $B$. Suppose that our last insertion was on $x_{2i-1}$ and our last removal was on $x_{2i}$. Let $C=A\backslash(B\cup\{x_{2i}\})\cup\{x_{2i-1}\}$ and $D=B\backslash\{x_{2i-1}\}\cup\{x_{2i}\}$. We have that $C\succ_2 D$ and $B\succeq_2 A\backslash B$, and it follows from the monotonicity of preferences that $C\cup\{x_{2k}\}\succeq_2 D$ and $B\cup\{x_{2k}\}\succeq_2 A\backslash B$. We will show that at least one of $D\cup\{x_{2k}\}\succeq_2 C$ and $A\backslash B\cup\{x_{2k}\}\succeq_2 B$ holds.

Suppose for contradiction that $C\succ_2 D\cup\{x_{2k}\}$ and $B\succ_2 A\backslash B\cup\{x_{2k}\}$. Then by the responsiveness of preferences, we have $C\succ_2 D\cup\{x_{2k}\}\succeq_2 B\succ_2 A\backslash B\cup\{x_{2k}\}\succ_2 C$, a contradiction. Hence at least one of $D\cup\{x_{2k}\}\succeq_2 C$ and $A\backslash B\cup\{x_{2k}\}\succeq_2 B$ holds.

In both Cases 1 and 2, we can find a subset $E\subseteq A$ of $k-1$ elements containing an item from each of the pairs $(x_1,x_2),(x_3,x_4),\ldots,(x_{2k-3},x_{2k-2})$ such that $E\cup\{x_{2k}\}\succeq_2 A\backslash E$ and $A\backslash E\cup\{x_{2k}\}\succeq_2 E$. This set $E$ will be crucial in the description of our algorithm.

We choose our set $T$ of $k+1$ items as follows. 
\begin{enumerate}
\item We include $x_{2k-1}$ and $x_{2k}$ in $T$. 
\item We include in $T$ either $E$ or $A\backslash E$ according to which one player 3 prefers. (If player 3 is indifferent between the two sets, we choose one of them arbitrarily.)
\end{enumerate}

Next, we claim that our chosen set $T$ is agreeable for all three players. We prove the claim separately for each of the players.
\begin{itemize}
\item From the perspective of player 1, the worst set $T$ we could have chosen is \[\{x_{2k-1},x_2,x_4,x_6,\ldots,x_{2k-2},x_{2k}\}.\] Since this set is agreeable to player 1 and preferences are responsive, our set $T$ is agreeable to player 1.
\item Since $E\cup\{x_{2k}\}\succeq_2 A\backslash E$ and $A\backslash E\cup\{x_{2k}\}\succeq_2 E$ and we include the remaining item $x_{2k-1}$, $T$ is also agreeable for player 2.
\item Since we choose the set $E$ or $A\backslash E$ that player 3 prefers and we include both of the remaining items $x_{2k-1}$ and $x_{2k}$, $T$ is also agreeable for player 3.
\end{itemize}
Hence $T$ is agreeable for all three players, as desired. This concludes the case $m$ even.

Assume now that $m=2k+1$ is odd. Let $S'$ be the set of all items of $S$ except $x_1$. We apply the algorithm from the case of $m$ even to choose a set $T\subseteq S'$ of $k+1$ items such that $T$ is agreeable for all three players, when the universe considered is $S'$. It follows that $T\cup\{x_1\}$ is a subset of $k+2$ items that is agreeable for all three players, when the universe considered is $S$.
\end{proof}

Even though the algorithm in Theorem \ref{thm:threeplayers} also relies on preferences on subsets, the number of subsets involved in the algorithm is only linear in the number of items, hence the algorithm remains efficient even when the number of items is relatively large.

One may ask whether we may restrict the algorithm in Theorem \ref{thm:threeplayers} to rely only on preferences on single items similarly to the case of two players. The following example shows that the answer is negative. In particular, a necessarily agreeable subset of size at most $\min\left(\left\lceil\frac{m+n-1}{2}\right\rceil,m\right)$ does not always exist.

\begin{example}
\label{ex:linear}
Suppose that there are six items $x_1,\ldots,x_6$, and consider the following preferences $\succeq_1,\succeq_2,\succeq_3$ of three players on the single items:
\[\succeq_1:x_1\succ x_4\succ x_5\succ x_6\succ x_2\succ x_3,\]
\[\succeq_2:x_2\succ x_5\succ x_6\succ x_4\succ x_3\succ x_1,\]
\[\succeq_3:x_3\succ x_6\succ x_4\succ x_5\succ x_1\succ x_2.\]
\end{example}

If we had access to the preferences on subsets, the algorithm in Theorem \ref{thm:threeplayers} would yield a subset of $\left\lceil\frac{6}{2}\right\rceil+1=4$ items that is agreeable to all three players. However, to obtain a necessarily agreeable subset, we need to include all of $x_1,x_2,x_3$ as they are most preferred by at least one player. Moreover, choosing only one of $x_4,x_5,x_6$ does not yield a necessarily agreeable subset for the player who ranks that item fourth. Hence a necessarily agreeable subset contains at least five items.

\section{Future Work}
\label{sec:conclusion}


Our work suggests a number of possible future directions. With the polynomial-time algorithms for two and three players in mind, a natural question to ask is whether we can similarly obtain efficient algorithms for more players. The algorithm for three players is already quite involved, so one might suspect that the problem is intractable for larger numbers of players. If that were to be the case, it would be useful to have a confirmation by means of a hardness result, even for some fixed large number of players. Since the problem is a search problem in which we know that a solution always exists, the problem would not be NP-hard, but could potentially be hard with respect to some subclass of TFNP such as PPAD or PLS. One could also ask about the NP-hardness of deciding the existence of an agreeable subset of some smaller size for which there is no guarantee of existence.

Since the algorithm in Theorem \ref{thm:threeplayers} relies on preferences on subsets, we can consider a different setting in which we are not given access to preferences on subsets but only to preferences on single items. In this setting, we need to choose a subset that is necessarily agreeable to all three players, and as Example \ref{ex:linear} shows, this may well force us to choose more items than we need to in the case where we have access to preferences on subsets. We can ask how many items we need to choose in the worst case, and whether there exists an efficient algorithm to compute those items. Of course, the same question can be asked for the corresponding setting with more players as well.

An interesting related question that goes beyond our setting is the following:  When is it possible to obtain an envy-free allocation between two (or more) parties if there are multiple players in each party, possibly with wildly differing preferences? To the best of our knowledge, the fair-division literature has so far only focused on settings in which each party consists of a single player. Whether we can come up with an algorithm along the lines of the undercut procedure \cite{BramsKiKl12} to find an envy-free allocation in the multiple-player setting when one exists is an appealing direction that we leave for future research.

\bibliographystyle{abbrv}
\bibliography{agreeable} 

\end{document}